\documentclass[letterpaper,11pt]{article}

\usepackage[normalem]{ulem}
\usepackage{color}
\usepackage{latexsym}

\usepackage{verbatim}

\usepackage{amsmath}
\usepackage{amssymb}
\usepackage{amsfonts}
\usepackage{graphicx}

\newtheorem{theorem}{Theorem}[section]
\newtheorem{lemma}[theorem]{Lemma}

\newtheorem{corollary}[theorem]{Corollary}

\newenvironment{proof}{{\bf Proof:\ }}{\hfill$\Box$\medskip}

\newcommand{\ignore}[1]{}
\newcommand{\remove}[1]{}

\newcommand{\etal}{{et al.\ }}
\newcommand{\fair}{{fair}}
\newcommand{\naive}{{na\"{i}ve}}

\newcommand{\WHILE}{{\tt while}}
\newcommand{\REPEAT}{{\tt repeat}}
\newcommand{\UNTIL}{{\tt until}}
\newcommand{\FOR}{{\tt for}}
\newcommand{\TO}{{\tt to}}
\newcommand{\IF}{{\tt if}}
\newcommand{\ELSE}{{\tt else}}
\newcommand{\RETURN}{{\tt return}}

\newcommand{\NOT}{{\tt not}}
\newcommand{\OR}{{\tt or}}

\newcommand{\link}{\mbox{{\tt link}}}
\newcommand{\cut}{\mbox{{\tt cut}}}
\newcommand{\meld}{\mbox{{\tt meld}}}
\newcommand{\hinsert}{\mbox{{\tt insert}}}
\newcommand{\hdelete}{\mbox{{\tt delete}}}
\newcommand{\deletemin}{\mbox{{\tt delete}-{\tt min}}}
\newcommand{\makeheap}{\mbox{{\tt make}-{\tt heap}}}
\newcommand{\makeitem}{\mbox{{\tt make}-{\tt item}}}

\newcommand{\decreasekey}{\mbox{{\tt decrease}-{\tt key}}}

\newcommand{\heads}{\mbox{{\tt heads}}}

\newcommand{\decreaseranks}{\mbox{{\tt decrease}-{\tt ranks}}}

\newcommand{\addchild}{\mbox{{\tt add}-{\tt child}}}

\newcommand{\decrement}{\mbox{{\tt decrement}}}

\newcommand{\findmin}{\mbox{{\tt find}-{\tt min}}}

\newcommand{\Makeheap}{\mbox{{\tt Make}-{\tt heap}}}
\newcommand{\Hinsert}{\mbox{{\tt Insert}}}
\newcommand{\Hdelete}{\mbox{{\tt Delete}}}
\newcommand{\Meld}{\mbox{{\tt Meld}}}
\newcommand{\Deletemin}{\mbox{{\tt Delete}-{\tt min}}}
\newcommand{\Findmin}{\mbox{{\tt Find}-{\tt min}}}
\newcommand{\Decreasekey}{\mbox{{\tt Decrease}-{\tt key}}}

\newcommand{\NULL}{\mbox{\it null}}
\newcommand{\info}{\mbox{\it info}}
\newcommand{\before}{\mbox{{\it before}}}
\newcommand{\after}{\mbox{{\it after}}}
\newcommand{\passive}{\mbox{{\it passive}}}
\newcommand{\rank}{\mbox{{\it rank}}}
\newcommand{\maxrank}{\mbox{{\it max\/}-{\it rank}}}

\newcommand{\marked}{\mbox{\it marked}}
\newcommand{\unmarked}{\mbox{\it unmarked}}

\newcommand{\kk}{k}

\newcommand{\ROOT}{\mbox{{\sl root}}}

\newenvironment{mytabbing}
  {\setlength{\topsep}{0pt}%
   \setlength{\partopsep}{5pt}%
   \tabbing}
  {\endtabbing}

\begin{document}

\title{Fibonacci Heaps Revisited}

\author{Haim Kaplan\thanks{Blavatnik School of
Computer Science, Tel Aviv University,
  Israel. Research supported by
The Israeli Centers of Research Excellence (I-CORE) program (Center
No. 4/11). E-mail: {\tt haimk@post.tau.ac.il}.}
\and Robert E. Tarjan\thanks{Department of Computer Science, Princeton University,
Princeton, NJ 08540 and Microsoft Research Silicon Valley, Mountain View, CA
94043.} \and Uri Zwick\thanks{Blavatnik School of
Computer Science, Tel Aviv University,
  Israel. Research supported by BSF grant no.\ 2012338 and by
The Israeli Centers of Research Excellence (I-CORE) program (Center
No. 4/11). E-mail: {\tt
    zwick@tau.ac.il}.}}

\date{}

\maketitle

\pagestyle{empty}
\thispagestyle{empty}

\begin{abstract}\noindent
The Fibonacci heap is a classic data structure that supports deletions in logarithmic amortized time and all other heap operations in $O(1)$ amortized time.  We explore the design space of this data structure.  We propose a version with the following improvements over the original: (i) Each heap is represented by a single heap-ordered tree, instead of a set of trees. (ii) Each decrease-key operation does only one cut and a cascade of rank changes, instead of doing a cascade of cuts. (iii) The outcomes of all comparisons done by the algorithm are explicitly represented in the data structure, so none are wasted.  We also give an example to show that without cascading cuts or rank changes, both the original data structure and the new version fail to have the desired efficiency, solving an open problem of Fredman.  Finally, we illustrate the richness of the design space by proposing several alternative ways to do cascading rank changes, including a  randomized strategy related to one previously proposed by Karger.  We leave the analysis of these alternatives as intriguing open problems.
\end{abstract}

\newpage

\section{Introduction}\label{sec:intro}

A \emph{heap} is a data structure consisting of a set of \emph{items}, each with a key selected from a totally ordered universe.  Heaps support the following operations:

\begin{description}
\item[$\makeheap( )$:] Return a new, empty heap.

\item[$\findmin(H):$] Return an item of minimum key in heap~$H$.  If~$H$ is empty, return null.

\item[$\hinsert(x, H)$:]
Return a heap formed from heap~$H$ by inserting
item~$x$, with predefined key. Item~$x$ must be in no heap.

\item[$\deletemin(H)$:]
Return a heap formed from non-empty heap $H$  by deleting
the item returned by\\ $\findmin(H)$.

\item[$\meld(H_1, H_2)$:] Return a heap containing all items in item-disjoint heaps $H_1$ and~$H_2$.

\item[$\decreasekey(x, v, H)$:] Given that~$x$ is an item in heap~$H$ with key no less than $v$,
return a heap formed from $H$ by changing
the key of~$x$ to~$v$.

\item[$\hdelete(x, H):$] Given that~$x$ is an item in heap~$H$,
return a heap formed by deleting
$x$ from~$H$.
\end{description}

The original heap~$H$ passed to $\hinsert$, $\deletemin$, $\decreasekey$, and $\hdelete$, and the heaps~$H_1$ and~$H_2$ passed to $\meld$, are destroyed by the operations.
Heaps do \emph{not} support search by key; operations \decreasekey\ and \hdelete\ are given the location of item~$x$ in heap~$H$ as part of the input.  The parameter~$H$ can be omitted from \decreasekey\ and \hdelete, but then to make \meld\ operations efficient one needs a separate disjoint set data structure to keep track of the partition of items into heaps.  For further discussion of this see \cite{KaShTa02b}.

\pagestyle{plain}
\setcounter{page}{1}

Fredman and Tarjan \cite{FrTa87} invented the \emph{Fibonacci heap}, an implementation of heaps that supports \deletemin\ and \hdelete\ on~$n$-item heaps in $O(\log n)$ amortized time and each of the other operations in $O(1)$ amortized time.  Applications of Fibonacci heaps include a fast implementation of Dijkstra's shortest path algorithm \cite{Di59,FrTa87} and fast algorithms for undirected and directed minimum spanning trees \cite{Edmonds67,GaGaSpTa86}.  Since the invention of Fibonacci heaps, a number of other heap implementations with the same amortized time bounds have been proposed \cite{Brodal96,BrLaTa12,Chan13,Elmasry10,HaSeTa0011,Hoyer95,KaTa08,Peterson87,Takaoka03}.  Notably, Brodal \cite{Brodal96} invented a very complicated heap implementation that achieves the time bounds of Fibonacci heaps in the worst case.  Brodal \etal \cite{BrLaTa12} later simplified this data structure, but it is still significantly more complicated than any of the amortized-efficient structures.  For further discussion of these and related results, see \cite{HaSeTa0011}.

In spite of the many competitors to Fibonacci heaps, the original data structure remains one of the simplest to describe and implement. We explore the design space of this data structure.
Our contributions are three. First we present a version of Fibonacci heaps in which each heap is represented by a single heap-ordered tree instead of a set of trees, each \decreasekey\ is implemented using cascading rank decreases instead of cascading cuts, and the outcome of each key comparison is explicitly represented in the data structure, so that no comparisons are wasted.  Second, we give an example to show that without cascading cuts or rank decreases, both the original data structure and ours fail to have the desired efficiency, solving an open problem of Fredman~\cite{Fredman05}.  Finally, to illustrate the richness of the design space we propose several alternative ways to do cascading rank decreases, including a  randomized method related to one previously proposed by Karger \cite{Kar06private}. We leave the analysis of these alternatives as intriguing open problems.


The remainder of our paper consists of six sections.  Section~\ref{sec:simple} describes our data structure.  Section~\ref{sec:analysis} analyzes it.  Section~\ref{sec:implementation} presents an implementation. Section~\ref{sec:cascading} gives an example showing that cascading is necessary to make Fibonacci heaps efficient, answering an open question of Fredman~\cite{Fredman05}. Section~\ref{sec:simpler} explores alternative ways to do cascading rank decreases.
Section \ref{sec:remarks} contains final remarks.

\section{Simple Fibonacci Heaps}\label{sec:simple}

We obtain our version of Fibonacci heaps by refining a well-known generic heap implementation.  We represent a heap by a rooted tree whose nodes are the heap items.  We access the tree via its root.  The tree is heap-ordered: each child has key no less than that of its parent.  Heap order implies that the root is an item of minimum key.  To do a \findmin\ on a heap, we return the root of its tree.

We do the update operations using a primitive called \emph{linking}.  Given the roots of two trees, we link them by comparing their keys and making the root of smaller key the parent of the root of the larger key, breaking a tie arbitrarily.  To make a heap, we create a new, empty tree.  To meld two heaps, if one is empty we return the other; if both are non-empty, we link the roots of their trees and return the resulting tree.  To insert an item into a heap, we make it into a one-node tree and meld this tree with the tree representing the heap.

We do \deletemin\ by repeated linking.  First, we delete the root of the tree representing the heap, making each of its children into the root of a separate tree.  If the root had no children, the heap is now empty.  Otherwise, we repeatedly link two roots until only one tree remains, and return this tree.

To decrease the key of an item~$x$ in a heap, we begin by replacing the key of~$x$.  If~$x$ is the root of its tree, this completes the operation.  Otherwise, we break the tree containing~$x$ into two trees, one containing all descendants of~$x$ (including $x$), the other containing the rest of the nodes in the tree.  We call this a \emph{cut} of~$x$.  We complete the operation by linking~$x$ with the root of the original tree.

To delete an arbitrary item~$x$ from its heap, we decrease its key to a value less than all other keys, and then do a \deletemin.

The efficiency of this implementation depends on how links are done in \deletemin.  (This is the only flexibility in the implementation.)  To keep the number of links small, we give each node $x$ a non-negative integer rank, denoted by $x.\rank$. We use ranks in a special kind of link called a \emph{\fair\ link}.  A \fair\ link of two roots can be done only if the roots have equal rank. Such a link increases the rank of the root of smaller key by 1 and makes this root the parent of the root of larger key, breaking a tie arbitrarily.  In contrast, a \emph{\naive\ link} ignores the ranks and merely makes the root of smaller key the parent of the root of larger key, without changing any ranks.

When linking two roots of the same rank, we have the choice of doing either a \fair\ link or a \naive\ link.  We do \fair\ links only when necessary to guarantee efficiency, namely only in \deletemin\ operations, since they require the extra step of changing a rank.

In addition to a rank, we give each node a state, either \emph{unmarked} or \emph{marked}.  We represent the state of a node $x$ by a Boolean variable $x.state$ that is true if and only if $x$ is marked. We use Boolean constants \emph{unmarked} and \emph{marked} whose values are \emph{false} and \emph{true}, respectively.
We refine the generic heap implementation as follows.  Each node added by an insertion has an initial rank of 0 and is initially unmarked.  Each link during an \hinsert\ or \meld\ is a \naive\ link.  During a \deletemin, as long as there are two roots of the same rank, we do a \fair\ link of two such roots.  Once all remaining roots have different ranks, we link them in any order by \naive\ links.
  When decreasing the key of an item~$x$ in a heap with root $h \ne x$, we unmark $h$ and then walk up the path in the tree from the parent of~$x$ through ancestors, unmarking each marked node reached and decreasing its rank by 1 if its rank is positive, until reaching an unmarked node.  We mark this node and decrease its rank by 1 if its rank is positive. Then we  cut~$x$,  and link~$x$ and $h$ via a \naive\ link.

To make the cascading rank change process completely transparent, we give a fragment of pseudo-code that implements it.  We denote the parent of a node~$x$ by $x.parent$.  The following code does the rank and state changes, assuming $x\ne h$:

\vspace*{-20pt}
\begin{center}
\parbox[t]{2.5in}{
\begin{tabbing}
aaa\=aaa\=aaa \kill
$h.state \gets \unmarked$ \\
$y \gets x$\\
\REPEAT:\\
\{\>  $y \gets y.parent$  \\
\>    \IF\ $y.\rank > 0$: $y.\rank \gets y.\rank - 1$ \\
\>    $y.state \gets \NOT\ y.state$ $\;\}$ \\
\UNTIL\ $y.state=\marked$ \\
\end{tabbing}
}
\end{center}

\vspace*{-20pt}
We call this data structure the \emph{simple Fibonacci heap}.
There are three major differences between it and the original version of Fibonacci heaps. First, the original does only \fair\ links, not \naive\ links. This makes each heap a set of trees rather than just one tree, accessed via a pointer to a root of minimum key.
Second, during the cascading process in a \decreasekey\ operation, the original method cuts each marked node reached, as well as unmarking it and decreasing its rank.  But except for cutting the node whose key decreases, these cuts are superfluous, as our analysis in the next section shows.
Third, our data structure maintains the outcomes of all key comparisons in the data structure itself, via fair and \naive\ links.  In the original, maintenance of a pointer to the root of minimum key requires comparisons among root keys whose outcomes are not explicitly maintained.  Also, when decreasing the key of $x$, $x$ is cut only if $x$ is not a root and its new key is less than that of its parent.  The outcome of this extra comparison is also not explicitly maintained.

\section{Analysis of Simple Fibonacci Heaps}\label{sec:analysis}

Our analysis uses the concept of \emph{active} and \emph{passive} children.  A root that becomes a child
by a \fair\ link is active; one that becomes a child by a  \naive\ link is passive.
  An active child that is marked and then unmarked becomes passive when it is unmarked; it remains passive until it becomes a root and then a child by a fair link once again.  We use active and passive children in the analysis only, not in the algorithm.

\begin{lemma} \label{lem:active}
For any node $x$, $x$ has at least $x.\rank$ active children.
\end{lemma}
\begin{proof}
The only way $x.\rank$ can increase is via a fair link, which adds an active child.  When $x$ loses an active child, either by a cut or because a child becomes unmarked and hence passive, $x.\rank$ decreases by~1 unless it is already 0.  The lemma follows by induction on time.
\end{proof}

Let the \emph{size} of a node $x$, denoted by $x.size$, be its number of descendants, including itself.  We prove that the rank of a node is at most logarithmic in its size.  Recall the definition of the Fibonacci numbers: $F_0 = 0$, $F_1 = 1$, $F_\kk  =
F_{\kk - 1} + F_{\kk - 2}$ for $\kk \ge 2$.  The Fibonacci numbers satisfy $F_{\kk + 2} \ge \phi^\kk$, where $\phi= (1 + \sqrt{5})/2$ is the \emph{golden ratio} \cite{Knuth98}.

\begin{lemma}\label{lem:rank}
Let $n_\kk$ be the minimum size of a node of rank at least $\kk$.  Then $n_\kk \ge \phi^\kk$.
\end{lemma}

\begin{proof}
Clearly $n_0 = 1$ and $n_1 = 2$.  Let $x$ be a node of rank at least $\kk\ge 2$ such that $x.size = n_\kk$.
By Lemma~\ref{lem:active}, $x$ has at least $\kk$ active children. Arrange these children in the order they were linked to $x$, most recent first. Fix $i < \kk$. Let $y$ be the $i^{\rm th}$ active child of $x$ in this order. We claim that just after $x$ acquires $y$ as a child, $x.\rank \ge \kk - i + 1$. To prove the claim, we observe that after $x$ acquires $y$ as a child, $x.\rank$ increases only when it acquires an active child, and of the active  children acquired after $y$, only $i - 1$  are currently active.  For  every other active child acquired after $y$, $x.\rank$ first increases and then decreases by one, since each such child either becomes passive or is cut by a \decreasekey\ operation.  Hence the net increase in $x.\rank$ from the time $y$ becomes a child of $x$ to the current time is at most $i - 1$. The claim follows.

The claim implies that $y.\rank \ge \kk - i$ when $y$ becomes a child of $x$.  Subsequently $y.\rank$ can decrease by at most 1, or $y$ would become passive.  Hence $y.\rank \ge \kk - i - 1$, which implies $y.size \ge n_{\kk - i - 1}$.  Applying this bound to each of the first $\kk - 1$ active children of $x$ and adding 2 for $x$ and its $\kk^{\rm th}$ active child, we obtain $n_\kk = x.size \ge n_{\kk - 2} + n_{\kk - 3} + \ldots + n_0 + 2$.  This implies $n_\kk \ge n_{\kk - 1} + n_{\kk - 2}$ for $\kk \ge 2$.  It follows by induction on $\kk$ that $n_\kk \ge F_{\kk + 2} \ge \phi^\kk$.
\end{proof}

\begin{corollary} \label{cor}
 A node of size $n$ has rank at most $\log_\phi n$.
\end{corollary}

To analyze the amortized efficiency of the heap operations, we use the standard potential-based framework \cite{Tarjan85}. We assign to each configuration of the data structure a real-valued {\em potential}. We define the amortized time of an operation to be its actual time (possibly scaled by a constant factor) plus the net increase in potential caused by the operation.  With this definition, the total time of a sequence of operations is equal to the sum of the amortized times of the operations plus the initial potential minus the final potential.  If the initial potential is zero and the final potential is non-negative, then the total actual time of the operations is at most the sum of their amortized times (times the scaling constant if any).

In our analysis we estimate the time of an operation other than a \deletemin, \decreasekey, or \hdelete\ to be 1; that of a \decreasekey\ to be $1 + \#iterations$, where $\#iterations$ is the number of iterations of the rank-decreasing loop; and that of a \deletemin\ to be $1 + \log_\phi n + \#links$, where~$n$ is the number of items in the heap and $\#links$ is the number of links, both \fair\ and \naive, done during the operation.  We treat a \hdelete\ as a \decreasekey\ followed by a \deletemin.  In the next section we present an implementation of the data structure that supports each operation except \deletemin, \decreasekey, and \hdelete\ in $O(1)$ actual time, each \decreasekey\ in $O(1 + \#iterations)$ actual time, and each \deletemin\ in $O(1 + \log_\phi n + \#links)$ actual time, justifying our estimates.

Let the \emph{degree} of a node $x$, denoted by $x.degree$, be the number of children of $x$. We use degrees in the analysis only; they are not part of the algorithm.
We define the potential of a node~$x$ to be $x.degree - x.\rank$,
plus $1$ if $x$ is a root, plus $2$ if $x$ is marked. For any node $x$, $x.degree \ge x.rank$ by Lemma~\ref{lem:active}, so every node has non-negative potential, and every root has potential at least $1$. We define the potential of a set of heaps to be the sum of the potentials of their nodes.  This potential is 0 initially (there are no nodes) and is always non-negative.  Thus in any sequence of heap operations the sum of the amortized times is an upper bound on the sum of their estimated times.

\begin{lemma}\label{lem:amortize}
The amortized time of a \makeheap, \findmin, or \meld\ is $1$, that of an \hinsert\ is $2$, that of a \decreasekey\ is at most $5$, and that of a \deletemin\ on a heap of $n$ items is at most $3\log_\phi n$.
\end{lemma}

\begin{proof}
A \naive\ link does not change the potential: the potential of the root that becomes a child decreases by 1 but that of the root that remains a root increases by 1 (its degree increases and its rank does not change).  A \fair\ link decreases the potential by 1: in this case the potential of the root that remains a root does not change, since both its degree and its rank increase by 1.  A \makeheap, \findmin, or \meld\ does not change the potential, so the amortized time of such an operation equals its estimated time of 1.  An insert creates a new root, of potential 1, making its amortized time 2 whether or not it does a \naive\ link.

A \decreasekey\ of a root does not change the potential and hence has an amortized time of 1.  Consider a \decreasekey\ of a non-root~$x$.
If the root is marked, unmarking it decreases the potential.
Each iteration of the rank-decreasing loop except the last  unmarks a node $y$ and may decrease its rank by 1, decreasing the potential by at least $2 - 1 = 1$. The last iteration marks a node and may decrease its rank by one, increasing the potential by at most 3. Cutting~$x$ from its parent does not change the potential: the potential of~$x$ increases by 1 since it becomes a root, but that of its parent decreases by 1, since its degree decreases by 1 and its rank does not change.  (Its rank may have decreased by 1 earlier, but we have already accounted for this.)  Combining these observations, we find that the \decreasekey\ increases the potential by at most $3 - (\#iterations-1)$, making the amortized time of the operation at most $(1+\#iterations)+(4-\#iterations)=\,5$.

Finally, consider a \deletemin\ on a heap of $n > 0$ items.  Let $h$ be the root of the tree representing the heap.  Deleting $h$ increases the potential by at most $h.\rank - 1$: deletion of $h$ reduces the potential by at least $1 + h.degree - h.\rank$, but each of the $h.degree$ new roots increases in potential by 1.
By Corollary~\ref{cor}, $h.\rank \le \log_\phi n$, so deleting $h$  increases the potential by
 at most $\log_\phi n -1$.  Each \fair\ link reduces the potential by 1. Also by Corollary~\ref{cor}, at most $\log_\phi n$ of the links are \naive, so at least $\#links - \log_\phi n$  are \fair.  We conclude that the entire \deletemin\ increases the potential by at most $\log_\phi n - 1 + \log_\phi n - \#links$, making the amortized time of the operation at most $(1+\log_\phi n + \#links) + (2\log_\phi n -1 - \#links)=\,3\log_\phi n$.
\end{proof}


\begin{theorem}\label{thm:amort}
The total estimated time of a sequence of heap operations starting with no heaps is $O(1)$ per operation other than \deletemin\ and \hdelete, and $O(\log n)$ per \deletemin\ and \hdelete, where~$n$ is the number of items currently in the heap.
\end{theorem}

\begin{proof}
The theorem is immediate from  Lemma~\ref{lem:amortize}.
\end{proof}


{\bf Remark:} The potential used in the proof of Lemma~\ref{lem:amortize} is like the original \cite{FrTa87} except that potential associated with children is moved to their parents.  This is necessary because Lemma~\ref{lem:active} is an inequality, not an equality.  Alternatively, we can define a child $y$ of a node $x$ to be \emph{hyperactive} if $y$ is one of the $x.\rank$ most recently linked active children of $x$.  Then the number of hyperactive children of a node equals its rank; and, if we give each node 1 unit of potential if it is not a hyperactive child, plus 2 if it is marked, then Lemma~\ref{lem:amortize} holds.  (The potential of a tree is the same under both definitions.)

\section{Implementation}\label{sec:implementation}

Our implementation of simple Fibonacci heaps mimics  the original implementation \cite{FrTa87}. We make each set of children a doubly-linked list: if~$x$ is a node, $x.\after$ and $x.\before$ are its next sibling and its previous sibling, respectively.  We need double linking to do cuts in $O(1)$ time. We also store with each node~$x$  pointers to its parent, $x.parent$, and to its first child on its list of children, $x.child$.  To find \fair\ links to do in \deletemin, we use a global array $A$ of nodes indexed by rank. We assume that the array is initially empty.  After deleting the root, we place each new root into the array location corresponding to its rank.  If this location is already occupied, we do a \fair\ link.  Once all remaining roots have been successfully stored in the array, we scan the array, linking pairs of roots by \naive\ links, until only one root remains, and leave the array empty again.  This implementation supports all the operations in the time bounds claimed in Section~\ref{sec:analysis}. Since the rank of a node increases or decreases by only 1 at a time, we can replace the global array by a global doubly linked list to obtain an implementation on a pointer machine. See~\cite{FrTa87}.
Figure~\ref{fig:pseudo-main} gives pseudo-code implementing the heap operations, as well as $\makeitem(\info,v)$, which creates a heap item having key~$v$ and  associated information~$\info$. Figure~\ref{fig:pseudo-auxiliary} gives pseudo-code implementing the auxiliary functions \link, \cut, \addchild, and \decreaseranks\ used in Figure~\ref{fig:pseudo-main}.

\newcommand{\MAKEHEAP}{
\parbox[t]{2.5in}{
\begin{mytabbing}
aaa\=aaa\=aaa \kill
$\makeheap()$: \\
\>      \RETURN\ $\NULL$
\end{mytabbing}
}}

\newcommand{\MAKEITEM}{
\parbox[t]{2.5in}{
\begin{mytabbing}
aaa\=aaa\=aaa \kill
$\makeitem(\info,v)$: \\
\{ \> $x \gets$ new node \\
\> $x.\info \gets \info$ \\
\> $x.key \gets v$ \\
\> $x.\rank \gets 0$ \\
\> $x.state \gets \unmarked$ \\
\> $x.child \gets \NULL$ \\
\>   \RETURN\ $x$ $\;\}$
\end{mytabbing}
}}


\newcommand{\INSERT}{
\parbox[t]{2.5in}{
\begin{mytabbing}
aaa\=aaa\=aaa \kill
$\hinsert(x, h)$: \\
\>      \RETURN\ $\meld(x, h)$
\end{mytabbing}
}}

\newcommand{\DELETEMIN}{
\parbox[t]{2.5in}{
\begin{mytabbing}
aaa\=aaa\=aaa\=aaa \kill
$\deletemin(h)$: \\
\{ \>    $x \gets h.child$ \\
\>      $\maxrank \gets 0$ \\
\>      \WHILE\ $x\ne\NULL$: \\
\>      \{ \>     $y \gets x$ \\
\>\>             $x \gets x.\after$ \\
\>\>             \WHILE\ $A[y.\rank]\ne \NULL$: \\
\>\>             \{ \>   $y \gets \link(y, A[y.\rank])$ \\
\>\>\>                   $A[y.\rank] \gets \NULL$ \\
\>\>\>                   $y.\rank \gets y.\rank + 1$ $\;\}$ \\
\>\>             $A[y.\rank] \gets y$ \\
\>\>             \IF\ $y.\rank > \maxrank: \maxrank \gets y.\rank$ $\;\}$ \\
\>      \FOR\ $i \gets 0$ \TO\ $\maxrank$: \\
\>      \{ \>    \IF\ $A[i] \ne \NULL$: \\
\>\>             \{ \>   \IF\ $x = \NULL$: $x \gets A[i]$ \\
\>\>\>                   \ELSE: $x \gets link(x, A[i])$ \\
\>\>\>                   $A[i] \gets \NULL$ $\;\}$ $\;\}$ \\
\>      \RETURN\ x $\;\}$
\end{mytabbing}
}}


\newcommand{\MELD}{
\parbox[t]{2.5in}{
\begin{mytabbing}
aaa\=aaa\=aaa \kill
$\meld(g, h)$: \\
\{ \>   \IF\ $g = \NULL$: \RETURN\ $h$ \\
\>      \IF\ $h = \NULL$: \RETURN\ $g$ \\
\>      \RETURN\ $\link(g, h)$ $\;\}$
\end{mytabbing}
}}

\newcommand{\DELETE}{
\parbox[t]{2.5in}{
\begin{mytabbing}
aaa\=aaa\=aaa \kill
$\hdelete(x, h)$: \\
\{ \>   $\decreasekey(x, -\infty, h)$ \\
\>      \RETURN\ $\deletemin(h)$ $\;\}$
\end{mytabbing}
}}

\newcommand{\DECREASEKEY}{
\parbox[t]{2.5in}{
\begin{mytabbing}
aaa\=aaa\=aaa \kill
$\decreasekey(x, v, h)$: \\
\{ \>    $x.key \gets v$ \\
\>       \IF\ $x = h$: \RETURN\ $h$ \\
\>       $h.state \gets \unmarked$ \\
\>       $\decreaseranks(x)$ \\
\>       $\cut(x)$ \\
\>       \RETURN\ $\link(x, h)$ $\;\}$
\end{mytabbing}
}}


\begin{figure}[t]
\begin{center}
\parbox[t]{2.5in}{
\MAKEITEM

\MAKEHEAP

\INSERT\MELD\DELETE}
\parbox[t]{2.5in}{
\DELETEMIN%
\DECREASEKEY
}
\end{center}
\vspace*{-10pt}
\caption{Implementation of the heap operations}
\label{fig:pseudo-main}
\end{figure}



\newcommand{\DECREASERANKS}{
\parbox[t]{2.5in}{
\begin{mytabbing}
aaa\=aaa\=aaa \kill
$\decreaseranks(y)$: \\
%
%
\{\> \REPEAT:\\
\> \{\>  $y \gets y.parent$  \\
\>\>    \IF\ $y.\rank > 0$: $y.\rank \gets y.\rank - 1$ \\
\>\>     $y.state \gets \NOT\ y.state$ $\;\}$ \\
\> \UNTIL\ $y.state=\marked$ $\;\}$
\end{mytabbing}
}}

\newcommand{\DECREMENT}{
\parbox[t]{2.5in}{
\begin{mytabbing}
aaa\=aaa\=aaa \kill
$\decrement(\kk)$: \\
\>       \IF\ $\kk > 0$: $\kk \gets \kk - 1$
\end{mytabbing}
}}


\newcommand{\LINK}{
\parbox[t]{2.5in}{
\begin{mytabbing}
aaa\=aaa\=aaa \kill
$\link(x, y)$: \\
\{ \>    \IF\ $x.key > y.key$: \\
\> \{\> $\addchild(x, y)$ \\
\>\> \RETURN\ $y$ $\;\}$ \\
\> \ELSE: \\
\> \{\> $\addchild(y, x)$ \\
\>\> \RETURN\ $x$ $\;\}\;\}$
\end{mytabbing}
}}



\newcommand{\ADDCHILD}{
\parbox[t]{2.5in}{
\begin{mytabbing}
aaa\=aaa\=aaa \kill
$\addchild(x, y)$: \\
\{ \>    $x.parent \gets y$ \\
\>       $z \gets y.child$ \\
\>       $x.\before \gets \NULL$ \\
\>       $x.\after \gets z$ \\
\>       \IF\ $z \ne \NULL$: $z.\before \gets x$ \\
\>       $y.child \gets x$ $\;\}$
\end{mytabbing}
}}


\newcommand{\CUT}{
\parbox[t]{2.5in}{
\begin{mytabbing}
aaa\=aaa\=aaa \kill
$\cut(x)$: \\
\{ \> $y \gets x.parent$ \\
\> \IF\ $y.child = x$: $y.child \gets x.\after$ \\
\>       \IF\ $x.\before \ne \NULL$: $x.\before.\after \gets x.\after$ \\
\>       \IF\ $x.\after \ne \NULL$: $x.\after.\before \gets x.\before$ $\;\}$
\end{mytabbing}
}}

\begin{figure}[t]
\begin{center}
\parbox[t]{2.5in}{
\LINK\ADDCHILD}
\parbox[t]{2.5in}{
\CUT\DECREASERANKS
}
\end{center}
\vspace*{-10pt}
\caption{\label{fig:pseudo-auxiliary}Auxiliary functions used in the implementation of the heap operations}
\end{figure}


There are several generic changes that can be made to our implementation, or to an implementation of any similar data structure for heaps. One change is to reduce the number of pointers per node from 4 to 3 by replacing the \emph{before} and \emph{after} pointers by one pointer whose meaning depends on the type of node: a first child points to the next sibling of its parent; a child other than the first points to its previous sibling.  (See \cite{FrSeSlTa86}.)  This trades time for space.  More complicated representations allow further space reductions.

A second change is to modify the specification of the heap operations so that they update old heaps instead of returning new ones.  One can do this by introducing separate heap objects.  In a minimal implementation of  this idea, a heap object $H$ has a single field $H.root$ that is the root of the tree representing the heap, or null if the heap is empty.  A less minimal implementation would store other useful information, such as the heap size, in each heap object.  Figure~\ref{fig:pseudo-alternative} gives an implementation using heap objects.  The operation $Insert(x, H)$ replaces $h \gets insert(x, h)$, and similarly for the other operations.

\newcommand{\MAKEHEAPC}{
\begin{mytabbing}
aaa\=aaa\=aaa \kill
$\Makeheap()$: \\
\{ \> $H\gets$ new heap \\
\> $H.root \gets \NULL$ \\
\>   \RETURN\ $H$ $\;\}$
\end{mytabbing}
}

\newcommand{\HINSERTC}{
\begin{mytabbing}
aaa\=aaa\=aaa \kill
$\Hinsert(x,H)$: \\
\> $H.root \gets \meld(x,H.root)$
\end{mytabbing}
}
\newcommand{\MELDC}{
\begin{mytabbing}
aaa\=aaa\=aaa \kill
$\Meld(H_1,H_2)$: \\
\> $H_1.root \gets \meld(H_1.root,H_2.root)$
\end{mytabbing}
}

\newcommand{\FINDMINC}{
\begin{mytabbing}
aaa\=aaa\=aaa \kill
$\Findmin(H)$: \\
 \>       \RETURN\ $H.\ROOT$
\end{mytabbing}
}

\newcommand{\DELETEMINC}{
\begin{mytabbing}
aaa\=aaa\=aaa \kill
$\Deletemin(H)$: \\
 \>        $H.\ROOT\gets \deletemin(H.root)$
\end{mytabbing}
}

\newcommand{\DECREASEKEYC}{
\begin{mytabbing}
aaa\=aaa\=aaa \kill
$\Decreasekey(H,x,k)$: \\
\>        $H.\ROOT \gets \decreasekey(x,k,H.\ROOT)$
\end{mytabbing}
}

\newcommand{\HDELETEC}{
\begin{mytabbing}
aaa\=aaa\=aaa \kill
$\Hdelete(H,x):$ \\
 \>   $H.\ROOT \gets \hdelete(x,H.root)$
\end{mytabbing}
}

\begin{figure}[t]
\begin{center}
\parbox[t]{2.5in}{
\MAKEHEAPC\MELDC\HINSERTC}
\hspace*{15pt}
\parbox[t]{2.5in}{
\FINDMINC\DELETEMINC\DECREASEKEYC\HDELETEC
}
\end{center}
\vspace*{-10pt}
\caption{\label{fig:pseudo-alternative}An implementation that modifies existing heaps.}
\end{figure}

A third change is to modify the implementation to do links lazily, only when needed to answer  \findmin\ queries. With this change, a heap consists of a set of heap-ordered trees, represented by a list of roots.  A \meld\ catenates the lists representing the two input heaps, without doing any links.  A \decreasekey\ of a non-root node $x$ adds to the list of roots instead of doing a link.  A \deletemin\ does a \findmin, which converts the list of roots into a single root, and then deletes the root, leaving its list of children as the new list of roots.  A \findmin\ links roots as in the implementation of \deletemin\ in Section~\ref{sec:simple}, but without first deleting a node; once all roots are linked, it returns the one remaining root.  The $O(1)$ time bound for \findmin\ becomes amortized instead of worst-case. As mentioned in Section \ref{sec:analysis}, the original version of Fibonacci heaps also represents the heap by a list of roots, but unlike the lazy implementation just described it maintains a pointer to a root of minimum key (making the time for \findmin\ $O(1)$ worst-case). Also, it does only \fair\ links, and it does all its links in \deletemin, not \findmin.

These three changes are independent of each other, so they can be combined arbitrarily.

We conclude this section by mentioning three heuristics that can be added to \decreasekey\ to reduce the number of iterations of the rank-decreasing loop in certain cases. These heuristics are independent and hence can be combined arbitrarily. Whether they provide improved performance is a question for experiments; none improves the $O(1)$ amortized time bound for \decreasekey, and they all complicate the implementation at least slightly, but they do preserve the analysis of Section~\ref{sec:analysis}. Furthermore, they provide ideas for actually simplifying the rank-decreasing process that we explore in Section~\ref{sec:simpler}.

The \emph{heap-order heuristic} cuts $x$ in $\decreasekey(x,v,h)$ only if $x$ is not a root \emph{and} its new key, $v$, is less than that of its parent, violating heap order.  As noted in Section \ref{sec:simple}, the original version of Fibonacci heaps uses this heuristic. A possible disadvantage of this heuristic is that the outcomes of such comparisons are not encoded in the data structure.

We can combine the two tests of whether $x$ is a child and whether $x$ violates heap order into a single test by defining the parent of a root to be itself.  Then $x.key < x.parent.key$ if and only if $x$ is a child that violates heap order.  Maintaining parents of roots requires initializing them in \makeitem\ and updating them in \cut.  In contrast, the implementation in Figures \ref{fig:pseudo-main} and \ref{fig:pseudo-auxiliary}
does not need to maintain parents of roots, since it never uses them.

The \emph{increasing-rank heuristic} stops the rank-decreasing loop when it reaches a node whose old rank is no less than that of its parent.  This heuristic reduces the worst-case time of \decreasekey\ to $O(\log n)$, since nodes visited in the rank-decreasing loop strictly increase in old rank.  To implement the heuristic, replace the implementations of \decreasekey\ and \decreaseranks\ by the implementations in Figure~\ref{fig:rank-increasing}.

\newcommand{\DECREASEKEYA}{
\parbox[t]{3in}{
\begin{mytabbing}
aaa\=aaa\=aaa \kill
$\decreasekey(x, v, h)$: \\
\{ \>    $x.key \gets v$ \\
\>       \IF\ $x = h$: \RETURN\ $h$ \\
\>       \IF\ $x.rank < x.parent.rank$: \\
\> \{\>       $h.state \gets \unmarked$ \\
\> \>        $\decreaseranks(x)$ $\;\}$ \\
\>       $\cut(x)$ \\
\>       \RETURN\ $\link(x, h)$ $\;\}$
\end{mytabbing}
}}

\newcommand{\DECREASERANKSA}{
\parbox[t]{3in}{
\begin{mytabbing}
aaa\=aaa\=aaa \kill
$\decreaseranks(y)$: \\
\{ \> \REPEAT:\\
\> \{\>  $y \gets y.parent$  \\
\>\>     $y.state \gets \NOT\ y.state$ \\
\>\>     $\kk \gets  y.rank$ \\
\>\>     $y.rank \gets  \kk-1$ $\;\}$ \\
\> \UNTIL\ $y.state=\marked$  \OR\ $\kk \ge y.parent.rank$ $\;\}$
\end{mytabbing}
}}

\begin{figure}[htb]
\begin{center}
\parbox[t]{3in}{
\DECREASEKEYA}
\parbox[t]{3in}{
\DECREASERANKSA}
\end{center}
\vspace*{-10pt}
\caption{\label{fig:rank-increasing} Implementation of \decreasekey\ using the increasing-rank heuristic.}
\end{figure}

The \emph{passive child heuristic} stops the rank-decreasing loop when it reaches a passive child (as defined in Section~\ref{sec:analysis}). To implement this heuristic, give each node one of three possible states: \emph{passive}, \emph{unmarked} (and active), or \emph{marked} (and active).
 Make roots implicitly passive.  Modify \link\ to make the new child passive, modify \deletemin\ to make each child added by a \fair\ link unmarked, and replace the implementations of \decreasekey\ and \decreaseranks\ by the implementations in Figure \ref{fig:passive-child}.

\newcommand{\DECREASEKEYP}{
\parbox[t]{3in}{
\begin{mytabbing}
aaa\=aaa\=aaa \kill
$\decreasekey(x, v, h)$: \\
\{ \>    $x.key \gets v$ \\
\>       \IF\ $x = h$: \RETURN\ $h$ \\
\>       $h.state \gets \passive$ \\
\>       \IF\ $x.state \not= \passive$: $\decreaseranks(x)$ \\
\>       $\cut(x)$ \\
\>       \RETURN\ $\link(x, h)$ $\;\}$
\end{mytabbing}
}}

\newcommand{\DECREASERANKSP}{
\parbox[t]{3in}{
\begin{mytabbing}
aaa\=aaa\=aaa \kill
$\decreaseranks(y)$: \\
\{ \> \WHILE\ $y.state=\marked$:\\
\> \{\>  $y.state \gets \passive$  \\
\>\>     $y.rank \gets y.rank-1$ \\
\>\>     $y \gets  y.parent$  $\;\}$ \\
\>    $y.rank \gets  y.rank-1$ \\
\> \IF\ $y.state=\unmarked$: $y.state \gets \marked$ $\;\}$
\end{mytabbing}
}}

\begin{figure}[htb]
\begin{center}
\parbox[t]{3in}{
\DECREASEKEYP}
\hspace*{15pt}
\parbox[t]{3in}{
\DECREASERANKSP}
\end{center}
\vspace*{-10pt}
\caption{\label{fig:passive-child} Implementation of \decreasekey\ using the passive child heuristic.}
\end{figure}


The passive child heuristic provides a bridge from the simple Fibonacci heaps of Section~\ref{sec:simple} to the original version \cite{FrTa87}. Indeed, we originally included this heuristic in our data structure, and then omitted it to obtain the version in Section~\ref{sec:simple}. Simple Fibonacci heaps with the heap-order and passive-child heuristics correspond exactly to the original version of Fibonacci heaps.  To obtain the data structure of the latter from that of the former, cut each passive child, producing a forest.  To map operations of the former to those of the latter, omit all the \naive\ links, and for each node~$y$ unmarked by \decreaseranks, cut~$y$.  With the potential function redefined as in the remark at the end of Section \ref{sec:analysis}, the analysis of the former corresponds exactly to that of the latter.

\section{Cascading is Required}\label{sec:cascading}

We have replaced the cascading cuts of Fibonacci heaps with cascading rank decreases.  Fredman \cite{Fredman05} has asked whether one can avoid cascading altogether.  Specifically, consider the following simplification of our data structure: keep ranks but eliminate node marking.  To decrease the key of a non-root~$x$, merely cut~$x$ from its parent, decrease the rank of the parent by 1 if its rank is positive, and link~$x$ and the original root by a \naive\ link.  We call this data structure a \emph{non-cascading heap}.  We show how to build a non-cascading heap of~$n$ nodes on which one \hinsert\ followed by one \deletemin\ reproduces the original structure, with the \deletemin\ taking $\Omega(n^{1/2})$ time.  Such pairs of \hinsert\ and \deletemin\ operations can be repeated indefinitely.  The number of operations needed to build the heap is $O(n^{3/2})$.  It follows that a sequence of~$m$ operations starting with no heaps can take $\Omega(m^{4/3})$ time.

Define~$S_\kk$ for a non-negative integer~$\kk$ to be a tree of $\kk+1$ nodes consisting of a root of rank~$\kk$ with~$\kk$ children. Define $T_\kk(i)$ for a positive integer $\kk$ and a non-negative integer $i\le \kk$ to be a tree whose root has~$\kk$ children that are the roots of $S_0,S_1,\ldots,S_\kk$, with $S_i$ missing. Thus $T_\kk(i)$ for $0<i<k$ consists of a root whose $\kk$ children are the roots of
$S_0,\ldots,S_{i-1},S_{i+1},\ldots, S_\kk$. See Figure~\ref{fig:Tji}. As a special case, $T_\kk(\kk)$ consists of a root whose $\kk$ children are the roots of $S_0, S_1,..., S_{\kk - 1}$.  Tree $T_\kk(\kk)$ has size $1 + \kk(\kk - 1)/2$. Given a heap whose tree is~$T_\kk(\kk)$, one \hinsert\ of an item with key larger than that of the root but smaller than that of all other nodes in the tree, followed by one \deletemin, will produce a new copy of~$T_\kk(\kk)$, the \deletemin\ having done~$\kk$ fair links and spent $\Omega(\kk)$ time.  Thus to get a bad example all we need to do is build a~$T_\kk(\kk)$ for large~$\kk$.

\begin{figure} [htbp]
\begin{center}
\includegraphics[scale=0.35]{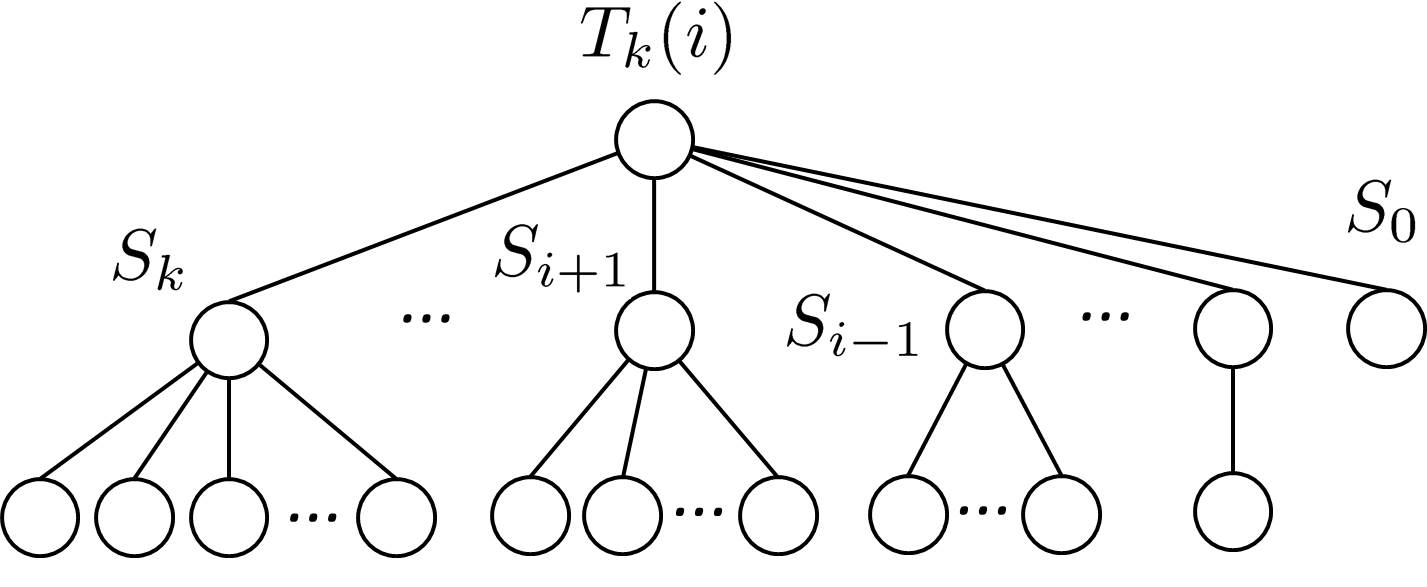}
\end{center}
\vspace*{-10pt}
\caption{The tree $T_\kk(i)$, for $1<i<k$.  \label{fig:Tji}}
\end{figure}

We can convert~$T_\kk(0)$ into~$T_{\kk + 1}(\kk + 1)$ by inserting a new item whose key is greater than that of the root. We can convert $T_\kk(i)$ with $i > 0$ into $T_\kk(i - 1)$ by doing two insertions, one \deletemin, and $i - 1$ \decreasekey\ operations, as follows. We insert two items whose keys are larger than that of the root. Then we do a \deletemin. After deletion of the original root, there are three roots of rank 0: the two newly inserted items and the root of the original~$S_0$. We link one of the newly inserted nodes with the root of the original~$S_0$ by a \fair\ link and then successively link the resulting tree with $S_1$, $S_2$, $\ldots$, $S_{i - 1}$ by \fair\ links.  We choose the keys so that among the roots linked, the root of $S_{i - 1}$ has smallest key and one of the newly inserted nodes has second-smallest key.  Then the tree formed by the links is $S_i$ but with one of its leaves having children that are the roots of $S_0$, $S_1$,$\ldots$, $S_{i - 2}$.  We then combine the remaining trees, which include one singleton, by \naive\ links that make the singleton the final root and all the other roots its children.  Finally we do $i - 1$ \decreasekey\ operations on the roots of $S_0$, $S_1$,$\ldots$, $S_{i - 2}$, making their roots into children of the final root.  This produces $T_\kk(i - 1)$. See Figure \ref{fig:Tji-toTji-1}.

\begin{figure} [htbp]
\begin{center}
\includegraphics[scale=0.35]{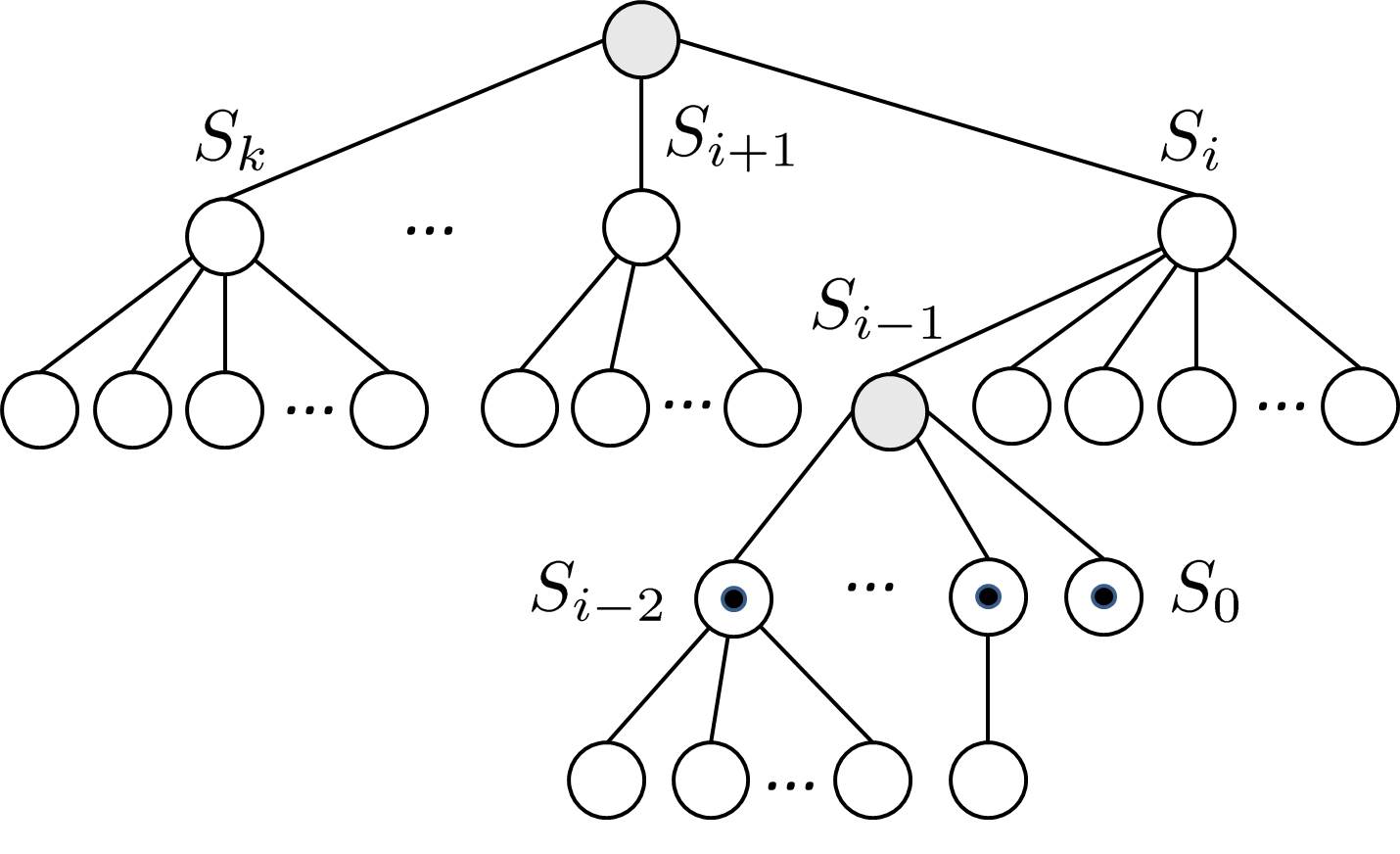}
\end{center}
\vspace*{-10pt}
\caption{The tree obtained when transforming $T_\kk(i)$ to $T_\kk(i-1)$, after two insertions and a \deletemin. The last step, which is not shown, is to do $i - 1$ \decreasekey\ operations on the roots of $S_0$, $S_1$,$\ldots$, $S_{i - 2}$ marked with a dot.
\label{fig:Tji-toTji-1}}
\end{figure}

By induction, we can convert $T_\kk(\kk)$ into $T_{\kk+1}(\kk+1)$ in $O(\kk^2)$ heap operations, and we can build
$T_\kk(\kk)$ from scratch in $O(\kk^3)$ heap operations.

\begin{theorem}
For any~$n$, there is a non-cascading heap of~$n$ nodes on which an \hinsert\ followed by a \deletemin\ can take $\Omega(n^{1/2})$ time and reproduce the original structure.  Such a heap can be built from scratch in $O(n^{3/2})$ operations.  For any~$m$ there is a sequence of $m$ operations starting from no heaps that takes $\Omega(m^{4/3})$ time.
\end{theorem}

\begin{proof}
The first two parts of the theorem are immediate from the construction above.  To prove the third part, assume $m\ge 6$. Let~$\kk$ be maximum such that~$T_\kk$ can be built in at most $m/3$ operations.  Then $\kk = \Omega(m^{1/3})$ by the construction above.  Build~$T_\kk$ in at most $m/3$ operations and then alternate \hinsert\ and \deletemin\ operations, with each pair recreating~$T_\kk$ and taking $\Omega(\kk)$ time.
\end{proof}

Essentially the same example shows that if cascading cutting is eliminated from the original version of Fibonacci heaps, the resulting data structure can take fractional-polynomial time per operation. This shows that the answer to Fredman's question is ``No''.

\vspace*{-5pt}
\section{Simpler Cascading?}\label{sec:simpler}

Given that cascading rank decreases are necessary, it is natural to ask whether there is a simpler way to decide when to do them.  We offer four possible methods, three deterministic and one randomized.  We cannot yet fully analyze any of these methods, and we leave doing so as intriguing open problems.

Two of our three deterministic methods simplify two of the heuristics in Section 4.  {\em Eager marking} simplifies the passive child heuristic by making marked nodes and active nodes identical.  Roots are implicitly unmarked.  A node that becomes a child by a fair link becomes marked.  The rank-decreasing loop stops when it reaches an unmarked node, which it leaves unmarked.  To implement this method, modify \link\ to make the new child unmarked, modify \deletemin\ to make each child added by a fair link marked, and replace the implementation of \decreaseranks\ by the implementation in Figure~\ref{fig:eager-marking}, given in the appendix.

Our other two deterministic methods eliminate marking entirely.  The {\em \naive\ increasing-rank method} stops the rank-decreasing loop when it reaches a node whose old rank is no less than that of its parent.  The {\em zero-rank method} stops the rank-decreasing loop when it reaches a node whose rank is~0.  To implement these methods, eliminate node states and replace the implementations of \decreasekey\ and \decreaseranks\ by those in Figures 9 and 10, respectively, given in the appendix.

It is straightforward to prove that with each of these methods, a node of size $n$ has rank at most $\lg n$.  It follows that the worst-case time of \decreasekey\ with the naive increasing-rank method is $O(\log n)$, and that the bounds of Theorem~\ref{thm:amort} hold for this method except that the amortized time per \decreasekey\ is $O(\log n)$.  For the other two methods we cannot prove even this weaker result.  We conjecture that Theorem~\ref{thm:amort} holds for none of these methods, but we have no counterexamples.

Another way to eliminate marking is to use {\em randomized cascading}, which at each iteration of the rank-decreasing loop stops with probability $1/2$.  To implement this method, eliminate node stated and replace the implementations of \decreasekey\ and \decreaseranks\ by those in Figure~\ref{fig:randomized}, given in the appendix.

Randomized cascading is a variant of a method proposed by Karger \cite{Kar06private}, {\em randomized cutting}, which applies to the original version of Fibonacci heaps \cite{FrTa87}.  It does not use marking; instead, at each iteration of the cascading-cut loop it stops with probability $1/2$.  With either randomized cascading or randomized cutting, \decreasekey\ takes $O(1)$ expected time.  In recent work independent of ours, Li and Peebles \cite{LiPeebles14} have given a tight analysis of randomized cutting.  They showed that the expected amortized time of \deletemin\ is $\Theta(\min \{n^{1/2}, (\log n\log s)/\log\log s \})$, where $s$ is the total number of \decreasekey\ operations.  Thus randomized cutting does not provide the efficiency of Fibonacci heaps.  Their lower bound example is  similar to our example in Section \ref{sec:cascading}.  Their upper bound analysis applies in a straightforward way to randomized cascading, but their lower bound does not (as far as we can tell).  Thus it remains open whether randomized cascading achieves the efficiency of Fibonacci heaps (in expectation).

%

\vspace*{-5pt}
\section{Remarks}\label{sec:remarks}

We have presented a one-tree version of Fibonacci heaps in which each \decreasekey\ requires only a single cut.  Ours is not the only heap structure with these properties.  For example, rank-pairing heaps \cite{HaSeTa0011} were specifically designed to need only one cut per \decreasekey, and there is a one-tree version of these heaps \cite{HaSeTa0011}.  But among all the known data structures that have the same amortized time bounds as Fibonacci heaps, we think ours is the simplest.  It has the additional nice property that the outcomes of all key comparisons are represented in the data structure (via links).  The pairing heap \cite{FrSeSlTa86}, a simpler self-adjusting structure, is faster in practice, but it does not support \decreasekey\ in $O(1)$ amortized time \cite{Fredman99}.  Indeed, obtaining a tight amortized time bound for \decreasekey\ in pairing heaps is an interesting open problem \cite{Fredman99,Pettie05}.

We have shown that without cascading rank changes, our data structure does not have the efficiency of Fibonacci heaps, solving an open problem of Fredman \cite{Fredman05}. We have proposed four simplified cascading methods. We leave the analysis of these methods as intriguing open problems.


\bibliographystyle{plain}

\makeatletter
\def\runninghead{\hrulefill\quad APPENDIX\quad\hrulefill}
\def\ps@headings{
\def\@oddhead{\footnotesize\rm\hfill\runninghead\hfill}}
\def\@evenhead{\@oddhead}
\def\@oddfoot{\rm\hfill\thepage\hfill}\def\@evenfoot{\@oddfoot}
\makeatother

\newpage
\setlength{\headsep}{15pt} \pagestyle{headings}

\appendix

\begin{figure}[t]
\begin{center}
\parbox[t]{2.5in}{
\begin{mytabbing}
aaa\=aaa\=aaa \kill
$\decreaseranks(y)$: \\
\> \WHILE\ $y.state=\marked$:\\
\> \{\>  $y.state \gets \unmarked$ \\
\>\>         $y \gets y.parent$ \\
\>\>         $y.rank \gets y.rank - 1$  $\;\}$
\end{mytabbing}
}
\end{center}
\vspace*{-10pt}
\caption{Implementation of \decreaseranks\ using eager marking.}
\label{fig:eager-marking}
\end{figure}

\begin{figure}[t]
\begin{center}
\parbox[t]{3in}{
\begin{mytabbing}
aaa\=aaa\=aaa \kill
$\decreasekey(x, v, h)$: \\
\{\>    $x.key \gets v$ \\
\>      \IF\ $x = h$: \RETURN\ $h$ \\
\>      \IF\ $x.\rank < x.parent.\rank$: $\decreaseranks(x, h)$ \\
\>      $\cut(x)$ \\
\>      \RETURN\ $\link(x, h)$ $\;\}$
\end{mytabbing}
}
\hspace*{10pt}
\parbox[t]{3in}{
\begin{mytabbing}
aaa\=aaa\=aaa \kill
$\decreaseranks(y, h)$: \\
\{\>    \REPEAT: \\
\> \{ \>    $y \gets y.parent$ \\
\>\>        $\kk \gets y.\rank$ \\
\>\>        $y.\rank \gets \kk - 1$ $\;\}$ \\
\>   \UNTIL\ $y = h$ or $\kk \ge y.parent.\rank$ $\;\}$
\end{mytabbing}
}
\end{center}
\vspace*{-10pt}
\caption{Implementation of \decreasekey\ using the \naive\ decreasing-rank method.}
\label{fig:naive-dec}
\end{figure}


\begin{figure}[t]
\begin{center}
\parbox[t]{3in}{
\begin{mytabbing}
aaa\=aaa\=aaa \kill
$\decreasekey(x, v, h)$: \\
\{\>    $x.key \gets v$ \\
\>      \IF\ $x = h$: \RETURN\ $h$ \\
\>      \IF\ $x.parent.\rank > 0$: $\decreaseranks(x, h)$ \\
\>      $\cut(x)$ \\
\>      \RETURN\ $\link(x, h)$ $\;\}$
\end{mytabbing}
}
\hspace*{10pt}
\parbox[t]{3in}{
\begin{mytabbing}
aaa\=aaa\=aaa \kill
$\decreaseranks(y, h)$: \\
\{\>    \REPEAT\: \\
\>\{\>    $y \gets y.parent$ \\
\>\>        $y.\rank \gets y.\rank - 1$ $\;\}$ \\
\>   \UNTIL\ $y = h$ or $y.parent.\rank = 0$ $\;\}$
\end{mytabbing}
}
\end{center}
\vspace*{-10pt}
\caption{Implementation of \decreasekey\ using the zero-rank method.}
\label{fig:zero-rank}
\end{figure}

\newcommand{\DECREASEKEYR}{
\parbox[t]{2.5in}{
\begin{mytabbing}
aaa\=aaa\=aaa \kill
$\decreasekey(x, v, h)$: \\
\{ \>    $x.key \gets v$ \\
\>       \IF\ $x = h$: \RETURN\ $h$ \\
\>      $\decreaseranks(x)$ \\
\>       $\cut(x)$ \\
\>       \RETURN\ $\link(x, h)$ $\;\}$
\end{mytabbing}
}}

\newcommand{\DECREASERANKSR}{
\parbox[t]{2.5in}{
\begin{mytabbing}
aaa\=aaa\=aaa \kill
$\decreaseranks(y)$: \\
\{ \> \REPEAT: \\
\> \{ \>        \IF\ $y.rank >0$: $y.rank\gets y.rank - 1$ \\
\>\>     $y \gets y.parent$  $\;\}$ \\
\>  \UNTIL\ $y=h$  \OR\ $\heads()$  $\;\}$
\end{mytabbing}
}}

\begin{figure}[t]
\begin{center}
\parbox[t]{2.5in}{
\DECREASEKEYR}
\hspace*{15pt}
\parbox[t]{2.5in}{
\DECREASERANKSR}
\end{center}
\vspace*{-10pt}
\caption{\label{fig:randomized}Implementation of \decreasekey\ using randomization.  Each call of function $\heads()$ returns a Boolean value that is true with probability $1/2$ and false with probability $1/2$, independent of the values returned by previous calls.}
\end{figure}

\end{document}